\documentclass[envcountsame]{llncs}
\usepackage{graphicx}
\usepackage{amssymb}
\usepackage{amsmath}
\usepackage{enumerate}
\usepackage{todonotes}
\usepackage[normalem]{ulem}
\usepackage{subfig}

\usepackage{color}
\newcommand{\remove}[1]{{}}

\usepackage{xcolor}
\colorlet{darkgreen}{green!40!black}

\newcommand{\ch}{\ensuremath\mbox{\textnormal{ch}}}
\newcommand{\chol}{\ensuremath\mbox{\textnormal{ch}}^{\mbox{\scriptsize{\textnormal{OL}}}}}
\title{List Colouring Big Graphs On-Line}% Alex suggested this new name
\titlerunning{List Colouring On-Line}  % abbreviated title (for running head)
\author{Martin Derka\thanks{The first author was supported by Vanier CGS.} \and Alejandro L\'{o}pez-Ortiz 
    \and Daniela Maftuleac}
\authorrunning{Martin Derka et al.} % abbreviated author list (for running head)
\institute{D.~Cheriton School of Computer Science,\\ University of Waterloo, Waterloo, ON, N2L 3G1, Canada\\
\email{\{mderka,alopez-o,dmaftule\}@uwaterloo.ca}
}

\begin{document}
\maketitle

\begin{abstract}
%On-line analysis of graphs plays a significant role in social
%networks.
In this paper, we investigate the problem of graph list
colouring in the on-line setting. We provide several results on
paintability of graphs in the model introduced by Schauz~\cite{cit:schauz}
and Zhu~\cite{cit:zhu}. We prove that the on-line version of Ohba's
conjecture is true in the class of planar graphs. We also consider
several alternate on-line list colouring models.
\end{abstract}

\section{Introduction}

Motivated by the processing of large social networks, we consider
the study of graph colouring in an on-line streaming manner. 
As we know, storing and analyzing a large social graph in main
memory for a single computer is not always possible. 
If the graph exceeds the capacity of the main memory, it has to be swapped 
to external memory. 
%Effectively, algorithm incurs costs associated
%with accessing vertices of the graph with detrimental impact on its performance.
Furthermore, in some cases, the social network graph to be processed may not be entirely
known in advance. Social networks rarely provide their complete graphs to
third parties due to privacy issues. Instead, they make their graphs accessible
via API. If an external application aims to analyze the entire graph, it has to employ
a local neighbourhood discovery protocol similar to web crawlers. In such a case,
the application incurs costs associated with accessing a vertex of the graph in the form
of the network communication necessary for issuing the API call.
Both of the aforementioned problems make the traditional off-line analysis of
social network graphs challenging and sometimes even infeasible. Thus, there is 
renewed interest in analysing graphs on-line.

Aside from their practical application, on-line graph algorithms have also been a 
rich source of theoretical problems with, for example, the celebrated theoretical
results of Schauz~\cite{cit:schauz} and Zhu~\cite{cit:zhu}.

In this paper, we investigate the graph list colouring problem in the on-line setting.
In list colouring, the vertices of a graph are pre-assigned lists of colours, and the task is to
properly colour the graph so that every vertex receives a colour from its list.
In what follows, we study the problem from a theoretical point of view and resolve
several open questions on this subject.

\section{Definitions and Previous Work}
The graphs considered in this paper are simple and undirected.
We follow the standard terminology of graph theory (cf.~for instance~\cite{cit:diestel}).

Let $G$ be a graph. A \emph{list assignment} is a function $L: V(G) \rightarrow 2^\mathbb{N}$
which assigns every vertex of $G$ a list of admissible colours. A proper colouring $c: V(G) \rightarrow \mathbb{N}$
is called an \emph{$L$-colouring} of $G$ if it assigns every vertex $v$ a colour $c(v)$ from its list $L(v)$.
The \emph{choosability number} of $G$, denoted by $\ch(G)$,
is the minimum number $k$ such that $G$ has an $L$-colouring whenever $L$ assigns every
vertex a list of size at least $k$. For any $k \geq \ch(G)$, graph $G$ is called \emph{$k$-choosable}.
The choosability number of a graphs is sometimes also called the \emph{list-chromatic number}.
 
The (off-line) list colouring problem---to decide whether a graph has an $L$-colouring---was introduced by Vizing in 1976~\cite{cit:brooks-list-rus}.
The choosability of graphs was investigated by Erd\"os, Rubin and Taylor~\cite{cit:erdos} and later by many others.
If $L$ assigns every vertex the same list of colours, the instance of the list colouring problem becomes
an instance of the ``standard'' vertex colouring problem. Thus, $$\chi(G) \leq \ch(G)$$ and the problem is NP-complete. Voigt~\cite{cit:voigt}
showed in 1993 that the choosability number $\ch(G)$ can be strictly larger than the chromatic number $\chi(G)$
even for planar graphs.

The list colouring problem was brought to the on-line setting independently by Schauz~\cite{cit:schauz}
and Zhu \cite{cit:zhu} in 2009. Both the authors formulated the problem as a game of
two players. In this paper, we follow the terminology of Schauz~\cite{cit:schauz}.

The game is played by two players called Mr.~Paint and Mrs.~Correct on a known graph $G$.
In each round, the first player, Mr.~Paint, takes a new colour $c$ and colours some (at least one)
uncoloured vertices. The colour $c$ cannot be used again. There are no restrictions on the
colouring of Mr.~Paint---he can colour two adjacent vertices with the same colour.
The other player, Mrs.~Correct, attempts to correct Mr.~Paint's mistakes. For this purpose,
she has a finite number of so-called \emph{erasers} assigned to every vertex. She can use an
eraser to remove the colour $c$ from any subset of vertices which were coloured by Mr.~Paint in this round.
An eraser can be used only once. By doing so, the number of erasers available for the given vertex
decreases. The game ends when the entire graph is properly coloured in which case Mrs.~Correct wins,
or when Mrs.~Correct cannot correct the colouring because she ran out of erasers for some vertex.
In such a case, Mr.~Paint wins. 

If $L$ is an assignment of erasers to the vertices of $G$
and Mrs.~Correct has a winning strategy leading to a proper colouring of $G$, the graph is called \emph{$L$-paintable}.
If $\ell \in \mathbb{N}$ is a number of erasers that need to be assigned to every vertex of $G$
for Mrs.~Correct to always have a winning strategy, the graph is called $(\ell+1)$-paintable.
The minimum such number $(\ell+1)$ is the \emph{paintability number} of a graph, and 
with respect to a graph $G$, it is denoted by $\chol(G)$.

Note that if Mr. Paint writes down all the colours suggested for each vertex into a list and we should 
decide who has the winning strategy, we get an instance of off-line list colouring.
Both Schauz~\cite{cit:schauz} and Zhu~\cite{cit:zhu} noted that if $G$ is not $k$-choosable, it
cannot be $k$-paintable.
So, the choosability number provides a lower bound on the paintability
number. We get that $$\chi(G) \leq \ch(G) \leq \chol(G).$$ Schauz~\cite{cit:schauz} provided an example of
a graph and an assignment of erasers $L$ where Mr.~Paint has a winning strategy, i.e., the graph is not
$L$-paintable, however it has an (off-line) list colouring for any list assignment with lists of the respective sizes~(see Appendix~\ref{ap:winning-strategy}).
Zhu~\cite{cit:zhu} proved that the complete bipartite graphs $K_{6,q}$ for $q \geq 9$ are not 3-paintable, 
however both $K_{6,9}$
and $K_{6,10}$ are 3-choosable. Thus, there are graphs with choosability strictly smaller than paintability.

In 1994, Thomassen~\cite{cit:thomassen} showed that all planar graphs are 5-choosable. 
Schauz~\cite{cit:schauz} adapted this technique to the on-line list colouring model 
to show that every planar graph $G$ is 5-paintable.
In the same paper, Schauz also noted that
\textit{``$\ell$-paintability is stronger than the $\ell$-list-colourability ($\ell$-choosability),
but not by much. Although [\ldots] there is a gap between these two notions, most theorems about list colourability
hold for paintability as well.''}

In~\cite{cit:ohba}, Ohba investigated the classes of graphs where the choosability number
equals the chromatic number. He showed that if a graph is sufficiently dense,
namely, if $|V(G)| \leq \chi(G)+\sqrt{2\chi(G)}$, then $\chi(G) = \ch(G)$. As a strengthening
of the result, he conjectured\footnote{
Ohba's conjecture was proved by Noel, Reed and Wu~\cite{cit:ohba-proof} in 2014.}
that if $G$ is a~graph with $|V(G)| \leq 2\chi(G)+1$ then $\chi(G) = \ch(G)$.

Kim et al.~\cite{cit:kim} studied Ohba's conjecture for multipartite graphs in the on-line setting. They pointed
out that, unlike the off-line case, graphs $K_{2\star(k-1),3}$ (the complete multipartite graphs
with $k-1$ parts of size 2 and one part of size 3) are not chromatic choosable on-line and thus
adjusted the inequality:

\begin{conjecture}[Ohba's On-Line Conjecture~\cite{cit:kim}]
\label{con:ohba-online}
Let $G$ be a graph with $|V(G)| \leq 2\chi(G)$. Then, $\chi(G) = \chol(G)$.
\end{conjecture}

A step towards proving Conjecture~\ref{con:ohba-online} was made by
Kozik, Micek and Zhu~\cite{cit:ohba-online}, who showed that it holds for the
graphs with independence number of at most~3. Furthermore, they proved~\cite{cit:ohba-online} that
the conjecture holds for graphs with $|V(G)| \leq \chi(G) + \sqrt{\chi(G)}$.  

Additionally, there are
various other results concerning the choosability and paintability of specific graph classes, see
e.g.~\cite{cit:kim,cit:huang,cit:kierstead,cit:ohba-online}.

\textbf{Our contribution.} The main result of this paper is a proof of Conjecture~\ref{con:ohba-online}
for the class of planar graphs (Section~\ref{sec:ohba}). We also prove several 
results about paintability of classes of sparse graphs (cf. Section~\ref{sec:game-model}).
Lastly we consider other possible
models for the on-line list colouring problem (Section~\ref{sec:models}).

\section{Classical Model}
\label{sec:game-model}
In this section, we focus on the ``classical'' game-theoretic model
of list colouring introduced by Schauz~\cite{cit:schauz}. We investigate
and extend results about paintability of graphs with small number of edges.
In order to do so, we work with the recursive definition of the on-line
list colouring problem: The game starts on graph $G$ with assignment of erasers $\ell$.
Once the players finish a round, i.e., Mr.~Paint colours a set of vertices $V_P$
and Mrs.~Correct erases the colours from some of them, denote her move by $V_C \subseteq V_P$, the vertices
in $V_P \setminus V_C$ that remained coloured can be removed from the graph---those vertices are properly
coloured and Mr.~Paint will never use the same colour again. So, the game proceeds on
a graph $G' = G[(V(G) \setminus V_P) \cup V_C]$ with one less eraser for the vertices in
$V_C$. See Appendix~\ref{ap:game} for formal definition.

Let us begin with the following observation (proof is provided in Appendix~\ref{ap:game}):

\begin{lemma}
\label{lem:good-degree}
Let $G$ be a graph, $v$ a vertex of degree $k$, and $\ell$ an
assignment of e\-ra\-sers 
such that $\ell(v) \geq k$.
If $G - v$ is $\ell$-paintable, then $G$ is $\ell$-pain\-ta\-ble.
\end{lemma}

\begin{theorem}
\label{thm:degeneracy}
Graphs with degeneracy $k \geq 0$ are $(k+1)$-paintable.
\end{theorem}

Theorem~\ref{thm:degeneracy} (see Appendix~\ref{ap:game} for the proof) states upper bounds for paintability of some
graph classes 
summarized by the following corollary:

\begin{corollary}
    \label{cor:apollonian}
 (a)~Forests are $2$-paintable.
(b)~Outer planar graphs are $3$-paintable.
(c)~Series-parallel graphs are $3$-paintable.
(d)~Apollonian networks are $4$-paintable.
(e)~$k$-regular graphs are $(k+1)$-paintable.
(f)~Planar graphs are $5$-paintable by inductive argument in~\cite{cit:schauz}. By degeneracy,
        they are trivially $6$-paintable.
\end{corollary}

\emph{Series-parallel} graphs are graphs with two distinguished vertices $s$ and $t$ called \emph{source} and
\emph{sink}. The class itself is defined inductively as follows: (1) an edge $(s,t)$ is a series-parallel graph;
and (2) any graph $G$ that can be obtained from two series-parallel graphs by a series or parallel composition
on theirs sources and sinks is a series-parallel graph.

It is easy to see that the class of series parallel graphs is a subclass of planar graphs and thus,
they are $5$-paintable. It is well-known that series-parallel graphs
are 2-degenerate, so they are even
$3$-paintable. 
We wish to offer an alternate \emph{inductive} argument which implies this statement.
The significance of this result is given by the following: while most of the techniques for list colouring off-line can
be transferred to the on-line setting, it is \emph{not always possible} for inductive arguments. 
Inspired by Thomassen's proof of 5-choosability of planar graphs~\cite{cit:thomassen}, and also by the proof of their 
5-paintability by Schauz~\cite{cit:schauz}, we prove a slightly stronger claim (the proof is provided in Appendix~\ref{ap:game}):

\begin{theorem}
\label{thm:sp}
Let $G=(V,E)$ be a series-parallel graph with source~$s$ and sink~$t$, and $\ell$ an assignment
of erasers such that $\ell(s) = 0$, $\ell(t) = 1$ and $\ell(v) = 2$ for any other vertex. Graph $G$ is $\ell$-paintable.
\end{theorem}

\section{Ohba's On-Line Conjecture}
\label{sec:ohba}

In this section, we prove the on-line version
of the Ohba's conjecture (cf.~Conjecture~\ref{con:ohba-online})
for planar graphs (see Theorem~\ref{thm:ohba}).
Our approach utilizes the following proposition which appears in~\cite{cit:carraher} (see also Appendix~\ref{ap:game}):

\begin{proposition}[Carraher et al.~\cite{cit:carraher}]
\label{prop:carraher}
If $G$ is a graph and
    $\ell$ an assignment of erasers to its vertices,
    the following holds for the game model of on-line list colouring.
    \begin{enumerate}[(a)]
    \item If $G$ is $\ell$-paintable, every subgraph $H$ of $G$ is $\ell$-paintable.
    \item If $\ell$ assigns every vertex $v$ of degree $k$ at least $k$ erasers,
    $G$ is $\ell$-paintable if and only if $G-v$ is $\ell$-paintable.
    \end{enumerate}
\end{proposition}

\begin{lemma}
\label{lem:connected}
Let $G$ be a graph and $\ell$ be an assignment of erasers such that Mr.~Paint has a winning strategy.
The winning strategy can be pursued by always selecting a set $P$ such that $G[P]$
is connected.
\end{lemma}
\begin{proof}
Let $P_1, P_2$ be two subsequent moves of Mr.~Paint
such that there are no edges between the vertices in $P_1$ and $P_2$,
and $C_1,C_2$ are arbitrary respective moves of Mrs.~Correct.
Denote the graph obtained by playing the moves by $H$. Observe that the graph obtained
by playing moves $P_1 \cup P_2$ and $C_1 \cup C_2$ by Mr.~Paint and Mrs.~Correct
is equal to $H$.

So, let $P'$ be a move in Mr.~Paint's winning strategy such that $G[P']$ is disconnected.
Let $H$ be a maximal connected component of $G[P]'$. Select $P_1 := V(H)$ and
$P_2 = P' \setminus P_1$, and replace the move $P'$ with two subsequent moves $P_1, P_2$ (in this order).
Assume that Mrs.~Correct has a winning strategy by moves $C_1, C_2$ after this modification. By the
observation above, Mrs.~Correct's response $C_1 \cup C_2$ is a winning response to Mr.~Paint's move $P'$.
This is a contradiction. Repeat this argument to produce a winning strategy of Mr.~Paint such that he
always colours a connected induced subgraph of $G$.
\qed
\end{proof}

\begin{theorem}
\label{thm:ohba}
Let $G$ be a planar graph with $|V(G)| \leq 2\chi(G)$. Then,
$\chi(G) = \chol(G)$.
\end{theorem}
\begin{proof}
Let $G$ be a connected planar graph with $|V(G)| \leq  2\chi(G)$.
Recall that if $G$ has independence number of at most $3$, the statment holds~\cite{cit:ohba-online}.
By the Four Colour Theorem, $\chi(G) \leq 4$, so we proceed in four cases based on $\chi(G)$.

\noindent\textbf{Case 1: $\chi(G) = 1$.} If the chromatic number
is $1$, the graph has no edges. Thus, Mrs.~Correct does not need
any erasers, and the graph is $1$-paintable.

\noindent\textbf{Case 2: $\chi(G) = 2$.} If the chromatic number
is $2$, the graph is bipartite. One should consider planar graphs
of size up to $4$ vertices. In fact, using Proposition~\ref{prop:carraher},
it is sufficient to prove the claim for the complete bipartite graphs
on $4$ vertices.
There are two\footnote{Note that $K_{0,4}$ is not connected, so it is both $1$-chromatic and $1$-paintable.} possible distributions of
vertices into the two partitions, so there are precisely two such graphs: $K_{1,3}$
and $K_{2,2} = C_4$. Graph $K_{1,3}$ is a tree, so it is $2$-paintable (cf.~Corollary~\ref{cor:apollonian}).

For $C_4 = (v_1,v_2,v_3,v_4)$, we consider three cases of Mr.~Paint's first move. If he colours one vertex only,
no erasers are used up and the game continues on a tree. So, Mrs.~Correct has winning strategy.
If Mr.~Paint colours more than two vertices, two of them, without loss of generality $v_1$ and $v_3$, are not adjacent to
each other. Mrs.~Correct leaves $v_1$ and $v_3$ coloured and uses erasers for the rest. Then the game continues
on a graph with two isolated vertices where no erasers are needed. So, the only option of Mr.~Paint is to
initially colour two adjacent vertices, say $v_1$ and $v_2$. Mrs.~Correct uses eraser for one of them, say $v_2$. The game continues
on a path $(v_2,v_3,v_4)$ where $v_2$ has no erasers, and the remaining vertices have each one eraser available.
One can easily see that Mrs.~Correct wins the game here too.

\noindent
\textbf{Case 3: $\chi(G) = 3$.} Applying Proposition~\ref{prop:carraher}, it is sufficient
to show that graphs with 6 vertices and chromatic number
$3$ need to be considered.
Graphs with no odd cycle are bipartite, thus $2$-chromatic. We
divide the case into two subcases: when $G$ contains a cycle of length $3$ and $5$.

If $G$ contains a cycle of length $5$, any independent set contains at most
two vertices of this cycle. Together with the last vertex, the independece
number of $G$ is at most $3$ and thus, the claim holds.

If $G$ contains a cycle $C$ of length $3$, any independent set contains at most
one vertex of this cycle. As $|V(G)| = 6$, the independence number is at most $4$,
in which case the vertices not in $C$, call them $u,v,w$, must be part of the independet set.
Assume that this is the case (otherwise the claim again holds).
Observe that $u,v,w$ are connected to at most two vertices of $C$ otherwise
they cannot be all members of the independent set. Hence, their degrees are at most
two, and $G$ is $2$-degenerate. Hence, $G$ is $3$-paintable by Theorem~\ref{thm:degeneracy}.

\noindent\textbf{Case 4: $\chi(G) = 4$.}
As $\chi(G)$ cannot be more than $4$, if one can prove
the claim for triangulated graphs on $8$ vertices, it holds for all the 4-chromatic
planar graphs on 8 vertices by Proposition~\ref{prop:carraher}.

Let $v$ be a vertex in $G$. As $G$ is triangulated, the neighbours $N(v)$ of $v$ form
a cycle $C$. The subgraph $C$ togehter with $v$ and its attachments to the vertices in $C$
is called a \emph{wheel}. Vertex $v$ is called a \emph{hub} of this wheel and $C$ is its \emph{rim}.
In order to show that the independence number of $G$ is at most $3$, we analyse $G$ based on
its wheels.

Observe that $G$, being a planar triangulation on $8$ vertices, contains precisely
$12$ triangular faces and $18$ edges. The maximum size of an independent set is at most $4$ as every
triangular face can contribute at most one vertex. A wheel of size $k$ contains
$2(k-1)$ edges. Furthermore, the maximum independence
number of such a wheel is $\lfloor \frac{k-1}{2}\rfloor$ and it cannot include
the hub (the only independent set which includes the hub has size $1$ as it cannot
include any other vertex).

The sum of vertex degrees in $G$ is $18 \cdot 2 = 36$. Hence, $G$ must contain
a vertex of degree at least $5$. If $G$ contains a vertex of degree $5$, the wheel
around this vertex has size $6$, the rim is a $5$-cycle and the wheel has $10$ edges.
Also, $G$ has two vertices $u,v$ that do not belong to this wheel.
In order to obtain an independent set of size $4$, both $u$ and $v$ have to be
added into an independent set $S$ of size $2$ found in the wheel. To fill the remaining
$8$ edges into $G$, at least one of $u,v$ must have degree at least $4$.
Thus, it has to be attached to one of the vertices that are already in $S$.
Hence, the independence number of $G$ is at most~$3$.

If $G$ contains a vertex $v$ of degree $6$, in order to construct an independent set of size $4$,
one has to find an independent set $S$ of size $3$ in the wheel around $v$ and fill in
the additional vertex $u$ that is not part of the wheel. Refer to Fig.~\ref{fig:ap-triangulation}.
Without loss of generality, on can select vertices $a,c,e$ into $S$. As $u$ needs to
belong to $S$ as well, it cannot be connected to neither of $a,c,e$. As $G$ is triangulated,
it must contain edges $(f,b), (b,d), (d,f)$ that enclose the wheels around $a,c,e$ respectively.
Then $u$ must be connected to $b,d,f$. Such a graph is an Apollonian network: starting with triangle
$b,d,f$, one subdivides its inner face by $v$ obtaining an embedded graph isomorphic to $K_4$,
and then subdivides all of its triangular faces by vertices $a,c,e,u$. By Corollary~\ref{cor:apollonian},
such a graph is $4$-paintable.

If $G$ contains a vertex of degree $7$, all the vertices in the graph form
a wheel of size $8$ around this vertex. Hence, the maximum independent set has size $3$.
So, in all subcases of Case 4, graph $G$ has independence at most $3$ and the claim holds, or it is the graph depicted
in Fig.~\ref{fig:ap-triangulation}, which is $4$-paintable by degeneracy. \qed
\end{proof}

\begin{figure}[tb]
\centering
\includegraphics[width=.35\textwidth]{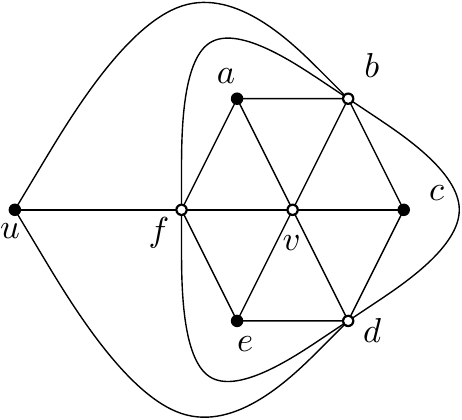}
\caption{Illustration for Case 4 of proof of Theorem~\ref{thm:ohba}. The
graph has independent set of size $4$, but it is $3$-degenerate, so $4$-paintable.}
\label{fig:ap-triangulation}
\end{figure}

\section{Different List Colouring Models}
\label{sec:models}
One of the initially suggested motivations for investigating the list colouring
problem on-line was analysing big graphs on a single server with limited
memory. A large disadvantage of the game model of Mr.~Paint and Mrs.~Correct is that it
does not contribute to lowering the memory requirements. In this section, we suggest
and investigate alternate models for the list colouring problem.

The memory requirements of the game %list colouring 
model are given by two facts:
\begin{itemize}
    \item the model requires that the graph is known to both the players; and
    \item the vertices are processed in clusters based on the colours that they share in their lists.
\end{itemize}
One could hope in eliminating the former requirement by giving a smart strategy
for Mrs.~Correct which will be ``good enough'' and yet will not require Mrs.~Correct
to know the entire graph. Unfortunately, there is no hope in overcoming the latter
obstacle. If a large number of vertices in the graph (possibly all of them) share
a single colour in their lists, they may have to be all processed (i.e., loaded into
the memory) within a single round of the game at some point. Moreover, there may
be multiple colours shared among many vertices, so there may be multiple memory-expensive
rounds within the game.

\subsection{True On-Line List Colouring Model}
An alternative to the ``classical'' game model for list colouring is what we call the \emph{true on-line list colouring model}.
In this model, the adversary reveals vertices of a graph $G$ in a sequence $\sigma = (v_1, v_2, v_3, \ldots, v_n)$ where every vertex $v_i$
is revealed with a list of admissible colours $L(v_i)$ and edges $(v_i,v_j)$ where
$j \leq i$. The on-line algorithm $\mathcal{A}$ has to assign a colour $c \in L(v_i)$ to every vertex $v_i$
immediately when the vertex appears. Additionally, it can re-assign a colour for any other vertex $v_j, j < i$.
At any given moment, the colouring of the graph maintained by $\mathcal{A}$ must be proper, i.e., no
two adjacent vertices are assigned the same colour. Assigning a colour to a vertex (both the initial assignment
and re-colouring) has an associated cost $1$. The \emph{cost} of $\mathcal{A}$ on sequence $\sigma$, denoted by $\mathcal{A}(\sigma)$ is the total
number of colour assignments that $\mathcal{A}$ performs in order to colour the graph. If
$G$ does not have an $L$-colouring, we set the cost $\mathcal{A}(\sigma) = \infty$.

The notion of cost in the true on-line list colouring model
naturally captures the cost of executing an on-line algorithm:
whenever a vertex needs to be coloured or re-coloured, it has
to be loaded into memory. Optimally, every vertex appears in the
memory exactly once, i.e., the algorithm will require no re-colourings.
The following theorem shows that even though this model exhibits
potential to reduce the memory traffic, it is not always the case.

\begin{theorem}
    \label{thm:no-competitive}
    There is no deterministic competitive on-line algorithm for the list
    colouring problem under the true on-line model.
\end{theorem}
\begin{proof}
Consider two disjoint paths $P_1$ and $P_2$ on $\frac{n}{2}$ vertices where
each vertex $v \in V(P_1) \cup V(P_2)$ is assigned a list $L(v)=\{1,2\}$.
Let $\mathcal{A}$ be the best deterministic on-line algorithm and
let the colour assigned to vertex $v$ by $\mathcal{A}$ be denoted
by $\mathcal{A}(v)$. A cruel adversary proceeds as follows.
If $P_1$ and $P_2$ have two ends, call them $x$ and $y$ respectively, such that
$\mathcal{A}(x) \neq \mathcal{A}(y)$, the adversary reveals a vertex
$z$ with list $\{1,2\}$ attached to $x$ and $y$. Let the resulting
path be denoted by $P$. At this point, one of $P_1$, $P_2$
has to be re-coloured, and the algorithm incurs cost $\frac{n}{2} + 1$.
If there are no two ends with distinct colours, the adversary
first reveals a vertex $z_0$ attached to one end of $P_1$ with list $\{1,2\}$
and then proceeds as in the previous case in order to join $P_1$ to $P_2$
(see Fig.~\ref{fig:construction}).
If the adversary constructed $P_1$ and $P_2$ in the same fashion (inductively starting
from isolated vertices), the on-line algorithm pays total cost
$$\mathcal{A}(\sigma_P) = \mathcal{A}(\sigma_{P_1}) + \mathcal{A}(\sigma_{P_2}) + \frac{n}{2} + q$$
where $\sigma_P, \sigma_{P_1}$ and $\sigma_{P_2}$ denote the sequences
of revealing $P, P_1$ and $P_2$ respectively and $q \in \{1,2\}$ is a constant
dependent of the way of connecting $P_1$ and $P_2$.
Applying the Master theorem~\cite{cit:cormen}, we obtain that $$\mathcal{A}(\sigma_P) \in \Theta(n\log n).$$

The best off-line algorithm knows the entire graph and thus it is
capable of colouring $P$ with cost $n$ (and in linear time). Hence,
there is no constant $c$ such that $\mathbb{\textsc{Opt}}(\sigma_P) \leq c \cdot \mathcal{A}(\sigma_P).$\qed
\end{proof}

\begin{figure}[tb]\centering
\begin{tabular}{cc}
   \includegraphics[trim=0mm 9mm 0mm 10mm, clip=true, width=0.45\textwidth]{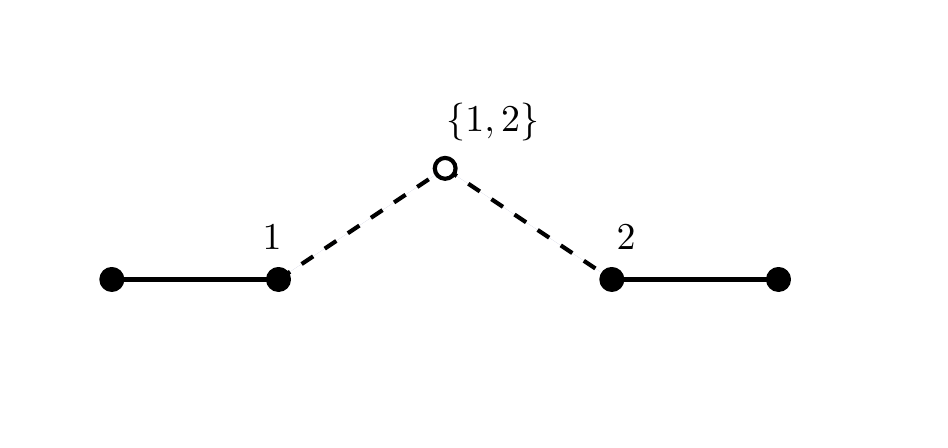}
   &
   \includegraphics[trim=0mm 9mm 0mm 10mm, clip=true, width=0.45\textwidth]{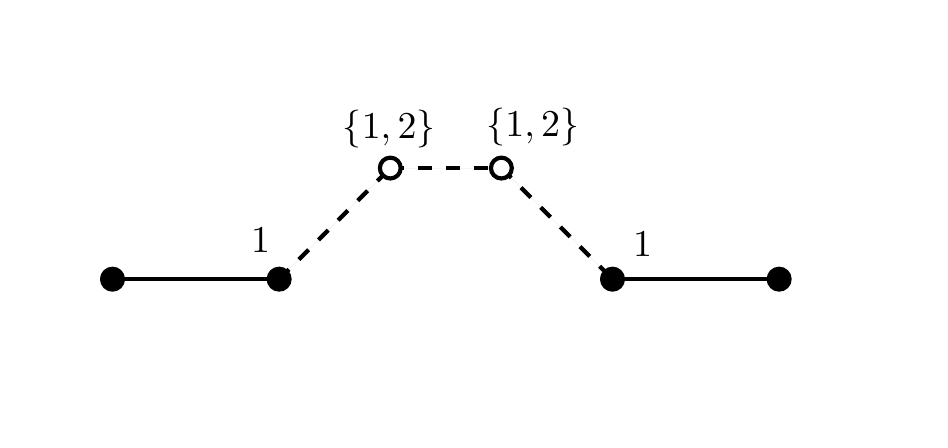}
   \\
   (a)
   &
   (b)
\end{tabular}
 \caption{Illustration for proof of Theorem~\ref{thm:no-competitive}.
(a) If two ends of two paths have distinct colours, the adversary connects them
with one vertex. (b) If two ends have the same colour, the adversary connects them
with two vertices.} \label{fig:construction}
\end{figure}

The proof of Theorem~\ref{thm:no-competitive} also provides a lower bound on
the worst case performance of an on-line algorithm under the true on-line
model for list colouring.

\begin{corollary}
    A lower bound on the worst case performance of any deterministic
    on-line algorithm in the true model for list colouring is $\Theta(n \log n)$
    where $n$ is the number of vertices in a graph. A lower bound on the
    competitive ratio beyond constant is $\Theta(\log n)$.
\end{corollary}

\subsection{Lazy Game Model}

Our endeavour to analyse models for on-line list colouring
includes attempts to improve the game model. In particular, we aim for reducing
the number of erasers that Mrs.~Correct needs for winning strategy.
In this respect, we considered a \emph{lazy game model}, where
Mrs.~Correct can decide to defer her move for later (i.e., she is temporarily allowed to
keep the graph coloured improperly). In this model,
Mr.~Paint suggests a set of vertices $S_i$ that he wants to colour
with colour $i$. If Mrs.~Correct defers her decision, Mr.~Paint
proceeds with suggesting additional set of vertices $S_{i+1}$ disjoint from $S_i$
that should be coloured with colour $i+1$.

\begin{theorem}
Let $G$ be a graph and $\ell$ be an assignment of erasers such that
$G$ is $\ell$-paintable (in the ``classical'' game model).
If Mrs.~Correct defers all her decisions and corrects
the colouring of Mr.~Paint only when all the vertices of the graph are coloured,
she has a winning strategy in the lazy game model.
\end{theorem}
\begin{proof}
It is easy to observe that for any strategy of Mr.~Paint, if Mrs.~Correct
defers all her decisions, eventually all the vertices of $G$ are coloured
by Mr.~Paint (this is not a proper colouring yet). Once all the vertices of $G$ are coloured,
Mrs.~Correct starts using erasers. Let us call this portion of the game a~\emph{phase}.
Once Mrs.~Correct erases colour from some of the vertices, the game proceeds
with another phase---Mr.~Paint colours sets of vertices, Mrs.~Correct defers all
her decisions, and uses erasers when all the vertices are coloured.

While Mr.~Paint colours vertices with colours $c_1, c_2, c_3, \ldots, c_\infty$, Mrs.~Correct will be treating
those colours as pairs $(p_i, c_j)$, where $p_i$ denotes that colour $c_j$ was used
by Mr.~Paint in $i$-th phase.

In preparation for the game, Mrs.~Correct can preprocess the graph. She is going
to construct a list assignment $L$ such that $$L(v) = \{1,2,3,\ldots,\ell(v)+1\}$$
for every vertex $v$, and compute an $L$-colouring $\mathcal{C}$ of $G$. Note that as
$G$ is $\ell$-paintable, the list colouring always exists. This colouring $\mathcal{C}$
will be now guiding lazy game of Mrs.~Correct.

The goal of Mrs.~Correct is to ensure that every vertex $v$ receives a colour $(p_i, c_j)$ such
that $i = \mathcal{C}(v)$. We need to argue that such an on-line colouring will both (a) be proper (with
respect to colours $c_1, c_2, c_3 \ldots, c_\infty)$; and (b) that Mrs.~Correct has enough erasers
to achieve this goal.

\begin{enumerate}[(a)]
\item 
A colouring is not proper if $G$ contains two adjacent
vertices with the same colour $c_i$. In the view of Mrs.~Correct, every two adjacent
vertices $u,v$ are coloured with colours $(p_{\mathcal{C}(u)},c_j)$ and $(p_{\mathcal{C}(v)}, c_k)$.
Since Mr.~Paint never re-uses a colour, no colour can be used in multiple phases.
So, if $c_j = c_k$, then ${\mathcal{C}(u)} = {\mathcal{C}(v)}$. However, since $\mathcal{C}$
is a proper colouring, it is impossible if $u,v$ are adjacent.
\item 
The goal is that Mrs.~Correct allows vertex $v$ to keep its colour in phase $p_{\mathcal{C}(v)}$.
This requires that $v$ has $\mathcal{C}(v)-1$ erasers available at the beginning of the
game. However, $\mathcal{C}(v) \leq \ell(v)+1$, thus $v$ requires at most $\ell(v)$ erases.
\qed
\end{enumerate}
\end{proof}

With respect to memory utilization, the lazy game model requires
loading multiple graph clusters at once. Thus, there is a clear
desire to minimize the number of deferrals that Mrs.~Correct makes.
The following theorem shows that there is no good bound on the number
of deferrals needed.

\begin{theorem}
There are graphs where Mrs.~Correct requires $\Theta(n)$ deferrals
where $n$ is the number of vertices in the graph.
\end{theorem}
\begin{proof}
We consider graph $G'$ obtained from the graph $G$ in Fig.~\ref{fig:schauz} (see Appendix~\ref{ap:winning-strategy})
by subdividing the edge $(v_5,v_6)$ by $n-8$ vertices $y_1,y_2,\ldots,y_{n-8}$ (Fig.~\ref{fig:th21}(a)).
Each vertex $y_i, 1 \leq i \leq n-8$, is assigned $1$~eraser.

In $G$, Mr.~Paint has a winning strategy.
In graph $G'$, using the ``classical'' game model, Mr.~Paint still has a winning strategy
if $n$ is even, since parities of all the cycles are preserved.
Consider the lazy game model for the graph $G'$. Every time Mrs.~Correct
decides to defer her turn, Mr.~Paint can choose to colour an uncoloured vertex $y_i$,
$1 \leq i\leq n-8$, without changing his winning strategy from the ``classical'' game model (the strategy is
described in Appendix~\ref{ap:winning-strategy}).
Thus, Mrs.~Correct requires at least $n-7$ deferrals, otherwise Mr.~Paint still has
a winning strategy.
\end{proof}

\subsection{Strong Game Model}

Aiming further for minimizing the number of erasers needed by
Mrs.~Correct, the following \emph{strong game model} appears to be very promising.
In the strong game model, Mrs.~Correct begins each round with selecting
a vertex that has to be included in Mr.~Paint's move. The rest of the vertices
suggested by Mr.~Paint can be arbitrary.
Such a model is in fact very natural from the implementation point of view:
an application can choose the order of iteration though colours (for example based
on a pre-processing of the graph).

In order to justify that this idea is worth future investigation, we
prove that the model discards Mr.~Paint's winning strategy on the
graph in Fig.~\ref{fig:schauz} (Appendix~\ref{ap:winning-strategy}).
A winning strategy for Mrs.~Correct,
 on the same graph, call it $G$, and assignment of erasers under the strong game model
is as follows:
\begin{figure}[tb]\centering
\begin{tabular}{ccc}
   \includegraphics[width=0.30\textwidth]{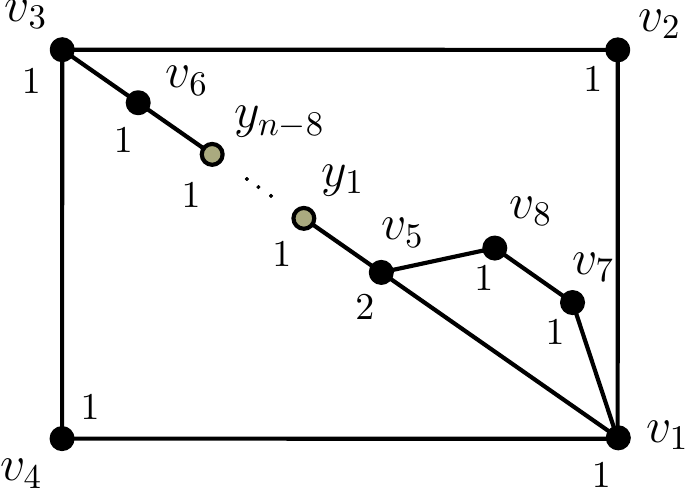} &\hspace{9mm} &
   \includegraphics[width=0.30\textwidth]{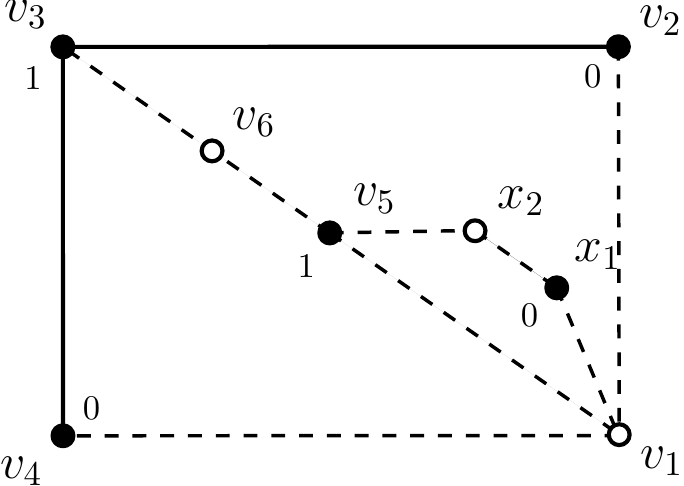}
   \\
  (a) & & (b) \\
\end{tabular}
\caption{(a)~A~graph with an assignment of erasers where Mrs.~Correct has a winning strategy if she defers $\Theta(n)$ decisions.
        Every vertex except for $v_5$ is assigned $1$~eraser; $v_5$ is assigned $2$~erasers.
(b)~Example of Mr.~Paint's first move $P_1 = \{v_1,v_2,v_4,v_5,v_6,x_1,x_2\}$ in the strong game.
After Mrs.~Correct uses erasers for $C_1 = \{x_1,v_2,v_4,v_5\}$, the game continues on a forest of trees.}
\label{fig:th21}
\label{fig:strong}
\end{figure}

At the beginning of the game, Mrs.~Correct chooses $v_1$ to be coloured
by Mr.~Paint in the first round. So, Mr.~Paint chooses set $P_1$ which includes
$v_1$ and perhaps some additional vertices. Note that by Lemma~\ref{lem:connected}, we can
assume that $G[P_1]$ is connected. Mrs.~Correct leaves $v_1$ coloured,
and uses an eraser for vertices in $P_1$ whose distance from $v_1$ is odd.
When the coloured vertices are removed, $G$ falls apart into
a forest of trees that can be rooted so that every non-leaf node
has at least one eraser available (see Fig.~\ref{fig:strong}(b)). It is easy to see that on such a graph,
Mrs.~Correct wins under the ``classical'' game model, and thus also under the strong
game model.

\section{Conclusions}
\label{sec:conclusions}
We considered the on-line graph list colouring problem under several
models. 
We extended the previous results about paintability of planar graphs
to some specific graph classes. 
We provided an inductive argument for $3$-paintability of series-parallel graphs. 
For future work, we would like to suggest 
extending the following two theorems to the on-line setting:

\begin{theorem}[Dvo\v{r}\'ak, Lidick\'y, \v{S}krekovski~\cite{cit:dvorak,cit:dvorak-5}] 
Any planar triangle-free graph without $4$-cycles adjacent to $4$- and $5$-cycles is $3$-choosable.
Any graph that can be drawn with at most two crossings is $5$-choosable.
\end{theorem}
The proofs of both the theorems are inductive and we believe that extension to the on-line setting is possible.

We also proposed two alternate models for on-line list colouring. We showed that both
the models differ from the ``classical'' game model of Mrs.~Correct and Mr.~Paint.

Our strong game model stands somewhere in between off-line list colouring 
and the ``classical'' on-line model. 
We showed that the model differs from the ``classical'' game model, but we do not believe that it
is equivalent to the list colouring problem off-line. 
It is known that traditional paintability is stronger than choosability. However, 
no example of a graph $G$ where $$\chol(G) > \ch(G)+1$$ is known (see e.q.~\cite[Question 4]{cit:regs}).
Therefore, in our opinion,
strong game paintability together with its relation to paintability and off-line choosability is very intriguing
and worth further investigation.

%\smallskip\noindent
%\textbf{Acknowledgements.}
%The authors would like to thank both the instructors
%for an interesting course and their engaging style of teaching.
% We would also like to thank them for the guidance and valuable suggestions
%provided throughout the course of this work.

% Add other sections as neccessary

%---------------------------- Bibliography -------------------------------

% Please add the contents of the .bbl file that you generate,  or add bibitem entries manually if you like.
% The entries should be in alphabetical order
%\begin{small}
%\bibliographystyle{abbrv}

%\bibliography{orthoballs}

\begin{thebibliography}{99}
%\input{bibliography}
\bibitem{cit:carraher}{
    J.~Carraher, S.~Loeb, T.~Mahoney, G.J.~Puleo, M.-T.~Tsai, D.B. West, \emph{Three topics in online list coloring}, accepted J. Combinatorics (2013).
}

\bibitem{cit:cormen}{
T.H.~Cormen, C.E.~Leiserson, R.L.~Rivest, C.~Stein, \emph{Introduction to Algorithms, Second Edition}, MIT Press and McGraw-Hill, 2001. ISBN 0-262-03293-7. 
Sections 4.3 (The master method) and 4.4 (Proof of the master theorem), pp. 73--90.}

\bibitem{cit:diestel}{
R.~Diestel, \emph{Graph Theory. Third Edition.}, Springer-Verlag, Berlin, 2005. 
}

\bibitem{cit:dvorak}{
Z.~Dvo\v{r}\'ak, B.~Lidick\'y, R.~\v{S}krekovski, 
\emph{$3$-Choosability 
of Triangle-Free Planar Gra\-phs with Constraints on $4$-Cycles},
SIAM J. Discr. Math. 24(3) (2010), 934--945.
}

\bibitem{cit:dvorak-5}{
Z.~Dvo\v{r}\'ak, B.~Lidick\'y, R.~\v{S}krekovski, 
\emph{Graphs with Two Crossings Are 5-Choosable},
SIAM J. Discrete Math. 25(4) (2011), 1746--1753.
}



\bibitem{cit:erdos}{
    P.~Erd\"os, A.L.~Rubin, H.~Taylor, \emph{Choosability in graphs}, Proc. West Coast Conference on Combinatorics, Graph Theory and Computing, Arcata,
    Congressus Numerantium 26 (1979), 125--157.
}
%\bibitem{cit:facebook}{
%  Facebook, {\em Facebook Reports Fourth Quarter and Full Year 2013 Results}, January 29, 2014, retrieved June~23, 2014.
%  Available online: \texttt{<http://investor.fb.com/releasedetail.cfm?\\ReleaseID=821954>}}
%\bibitem{cit:twitter}{
%  P. Gupta, A. Goel, J. Lin, A. Shama, D. Wang, R. Zadeh, {\em WTF: the who to follow service at Twitter},
%  Proc. of the 22nd international conference on World Wide Web (WWW13) (2013), pp. 505--514.
%}
%\bibitem{cit:kral}{
%	J.~Hladk\'y, D.~Kr\'al', U.~Schauz, \emph{Brooks' Theorem via the Alon-Tarsi Theorem}, Discrete Mathematics 310 (2010) 3426--3428.
%	doi: 10.1016/j.disc.2010.07.019
%}

%\bibitem{cit:hoffman-johnson}{
%D.G.~Hoffman and P.D.~Johnson~Jr., \emph{On the choice number of $K_{m,n}$}, Proc. 24th SE Intl. Conf. Comb., 
%Graph Theory, and Computing (Boca Raton, FL, 1993), Congr. Numer. 98 (1993), 105--111.
%}

\bibitem{cit:huang}{
P.~Huang, T.~Wong, X.~Zhu, \emph{Application of polynomial method to on-line colouring of
graphs}, European J. Combin. (2011)
}

\bibitem{cit:kim}{
S.-J.~Kim, Y.~Kwon, D.D.~Liu and X.~Zhu, \emph{On-line list colouring of complete multipartite
graphs},  Electron. J.Combin. 19 (2012), Paper \#P41, 13 pages.
}

\bibitem{cit:kierstead}{
	H.A.~Kierstead, \emph{On the choosability of complete multipartite graphs with part size three}, Discrete
Math. 211 (2000), 255--259.
}

%\bibitem{cit:brooks-list-eng}{
%	A.V.~Kostochka, M.~Stiebitz, B.~Wirth, \emph{The colour theorems of Brooks and Gallai extended}, Discrete Math. 162 (1996) 299--303.
%}
\bibitem{cit:ohba-online}{
	J.~Kozik, P.~Micek, X.~Zhu, \emph{Towards an on-line version of Ohba's conjecture},
	European Journal of Combinatorics 36 (2014) 110--121.
}
%\bibitem{cit:nauty}{
%    B.D.~McKay, A.~Piperno, \emph{Practical Graph Isomorphism II}, J. Symbolic Computation 60 (2013) 94-112.
%}

\bibitem{cit:ohba-proof}{
	J.~Noel, B.~Reed, H.~Wu, \emph{A Proof of a Conjecture of Ohba}, manuscript to appear, 2014, arXiv:1211.1999v2.
}
\bibitem{cit:ohba}{
	K.~Ohba, \emph{On Chromatic-Choosable Graphs}, J. Graph Th. 40(2) (2002), 130--135.
}
%\bibitem{cit:4ct}{
%	N. Robertson, D.P.~Sanders, P.~Seymour, R.~Thomas, \emph{The Four-ColourTheorem},  J. Combin. Theory Ser. B 70(1) (1997), 2--44.
%}
\bibitem{cit:schauz}{
	U.~Schauz, \emph{Mr. Paint and Mrs. Correct}, Electron. J.~Combin. 16 (2009) \#R77.
}
\bibitem{cit:schauz-alon-tarsi}{
	U.~Schauz, \emph{Flexible color lists in Alon and Tarsi's Theorem, and time scheduling with unreliable participants}, Electron. J. Combin. 17 (2010) \#R13, 18 p.
}
\bibitem{cit:regs}{
U.~Schauz, X.~Zhu, T.~Mahoney, \emph{REGS in Combinatorics 2011}, University of Illionis. Available online: 
{\tt <http://www.math.uiuc.edu/{\textasciitilde}west/regs/paint.htm>}}

%\bibitem{cit:sage}{
%    W.A.~Stein et~al., \emph{{S}age {M}athematics {S}oftware ({V}ersion 6.2)}, The Sage Development Team, 2014, {\tt http://www.sagemath.org}.
%}
\bibitem{cit:thomassen}{
	K.~Thomassen, \emph{Every Planar Graph Is 5-Choosable}, Journal of Combinatorial Theory, Series B Vol 62(1) (1994), 180--181.
}
\bibitem{cit:brooks-list-rus}{
	V.G.~Vizing, \emph{Colouring the vertices of a graph with prescribed colours}, Metody Diskret. Anal. v Teorii Kodov i Shem 29 (1976) 3--10 (in Russian).
}
\bibitem{cit:voigt}{
    M.~Voigt, \emph{List colourings of planar graphs}, Discrete Math. vol. 120 (1993), 215–-219.
}
\bibitem{cit:zhu}{
    X.~Zhu, \emph{On-line list colouring of graphs}, Electron. J. Combin. 16(1) (2009) \#R127.
}
\end{thebibliography}
%\end{small}

\newpage
\appendix

\section{Difference between Choosability and Paintability}
\label{ap:winning-strategy}

\begin{figure}\centering
   \includegraphics[width=0.29\textwidth]{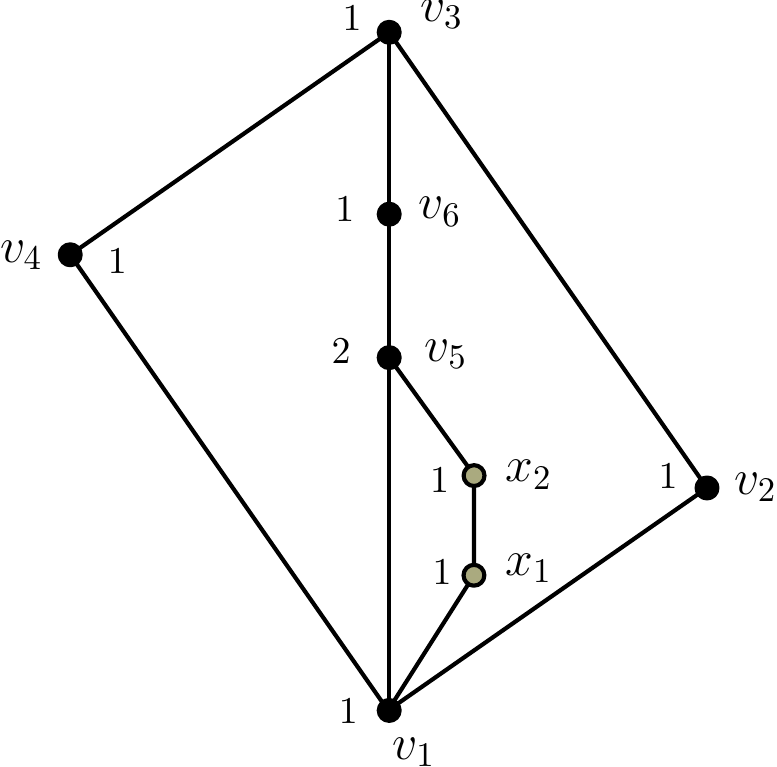}
\caption{A~graph with an assignment of erasers where Mr.~Paint has a winning strategy~\cite{cit:schauz}
            even though a respective (off-line) list colouring always exists. Every vertex except
        for $v_5$ is assigned $1$~eraser; $v_5$ is assigned $2$~erasers.}
\label{fig:schauz}
\end{figure}

The fact the list colouring on-line differs from list colouring off-line can be justified 
for instance on the graph depicted in Fig.~\ref{fig:schauz}. Consider an arbitrary 
list assignment that assigns every vertex except for $v_5$ a list of size $2$, and vertex $v_5$ a list
of size $3$. As proved in~\cite{cit:schauz}, the graph an (off-line) list colouring for any such $L$.

Consider analogous assignment of erasers $\ell$ which assigns every vertex
in $G$ except for $v_5$ one eraser. Vertex $v_5$ receives two erasers.
Mr.~Paint's winning strategy is as follows. He starts
by colouring with the colour $c_1$ the vertices $x_1$ and $x_2$.
Now, Mrs.~Correct has two options:
\begin{enumerate}
\item \textit{Mrs.~Correct chooses to use the only eraser of $x_1$.}
Mr.~Paint colours the following sets of vertices at each of his turns:
first $x_1$ and $v_1$ using colour $c_2$, then $v_1$, $v_4$ and $v_2$
using colour $c_3$, then $v_4$ and $v_3$ using colour $c_4$,
then finally $v_2$ and $v_3$ using colour $c_5$.
\item \textit{Mrs.~Correct chooses to use the only eraser of $x_2$.}
Mr.~Paint wins by assigning first the same colour to $x_2$ and $v_5$, this
way neutralizing one of the erasers of $v_5$, and next assigning
the same colour to the odd cycle $(v_1,v_4,v_3,v_6,v_5)$ where
each vertex has only one eraser left.
\end{enumerate}

\section{Classical Model}
\label{ap:game}

\begin{definition}[Paintability---reformulation~\cite{cit:schauz}]
For an assignment of era\-sers $\ell \in \mathbb{N}^V$ and a graph $G = (V,E)$,
\emph{paintability} is recursively defined as follows:
    \begin{enumerate}[(i)]
    \item $G = \emptyset$ is $\ell$-paintable (since $V = \emptyset$, $\ell$ is an empty tuple).
    \item $G \neq \emptyset$ is $\ell$-paintable if $\ell(v) \geq 1$ for every $v \in V$, and every non-empty subset $V_P \subseteq V$
    of vertices contains a \emph{good} subset $V_C \subseteq V_P$ such that $G[(V \setminus V_P) \cup V_C]$
    is $\ell'$-paintable where $\ell'(v) = \ell(v) - 1$ for $v \in V_C$ and $\ell'(v) = \ell(v)$ otherwise.
    \end{enumerate}
\end{definition}

\begin{proof}[of Theorem~\ref{thm:degeneracy}]
Let $k \geq 0$ be arbitrary, but fixed.
Proceed by induction on the number of vertices $n$. If $n = 1$,
the graph has one isolated vertex and thus, it is $(k+1)$-paintable.
Let $G$ be a $k$-degenerate graph with $n > 1$. By definition, it
has vertex $v$ of degree at most $k$ such that $G - v$ is also a
$k$-degenerate graph. So, by induction, $G-v$ is $(k+1)$-paintable
and $G$ is $(k+1)$-paintable
by Lemma~\ref{lem:good-degree}. \qed
\end{proof}

\begin{proof}[of Theorem~\ref{thm:sp}]
By induction on the number, we prove the following claim:
Graph $G' = G - (s,t)$ (if edge $(s,t)$ does not exist, $G' = G$) with $n \leq k$ is $\ell'$-paintable for
$\ell'(s) = \ell'(t) = 0$ and $\ell'(v) = \ell(v) = 2$ for any other vertex $v$. Note that the claim immediately implies
$\ell$-paintability of $G$ if it does not contain edge $(s,t)$. If edge $(s,t)$ is present, Mrs.~Correct
can follow the same strategy as in $G'$, but use the one eraser assigned to $t$ by $\ell$ in order
to ensure different colours of $s$ and $t$. Thus, $\ell$-paintability of $G$ is guaranteed in this case as well.

In the base case for $n = 2$, $G$ is an edge $(s,t)$ and $G'$ if formed by two isolated vertices.
Thus, Mrs.~Correct does not need any erasers and $G'$ is $\ell'$-paintable.

So, let $G$ be a series parallel graph with $k+1$ vertices. By definition, $G$ is obtained
from two series-parallel graphs, call them $G_1,G_2$ with sources $s_1,s_2$ and sinks $t_1,t_2$
respectively. We proceed in two cases:

\smallskip
\noindent \textit{Case 1: $G$ is obtained by parallel composition.}
By definition of the parallel composition $s = s_1 = s_2$ and $t = t_1 = t_2$. Let Mr.~Paint's
move be set of vertices $V_P$. By the recursive definition of paintability, the task of Mrs.~Correct
is to find a good subset $V_C$ where erasers will be used such that $G[(V \setminus V_P) \cup V_C]$ is
paintable for the respective assignment of erasers (the number of erasers for vertices in $C$ decreases).
By induction, this is possible for both graphs $G_1-(s,t), G_2-(s,t)$ with assignment of erasers $\ell'$, and Mr.~Paint's moves $V(G_1) \cap V_P$
and $G_2 \cap V_P$. Denote the good subsets played by Mrs.~Correct $V_{C_1}$ and $V_{C_2}$ respectively.
Note that Mrs.~Correct cannot include $s$ and $t$ in her good subset in any of the graphs  as they do not have any erasers. Thus,
set $V_C = V_{C_1} \cup V_{C_2} (\subseteq V_P)$ is good for $G-(s,t)$ and it is $\ell'$-paintable. If $G$ does not
contain edge $(s,t)$, it is $\ell$-paintable immediately, otherwise the aforementioned argument applies.

\smallskip
\noindent \textit{Case 2: $G$ is obtained by series composition.}
Note that in this case, $G$ does not contain edge $(s,t)$, so $G = G'$. Furthermore,
$t_1 = s_2$ is assigned 2 erasers by $\ell$.
Again, we show that Mrs.~Correct can find a good subset $V_C$ in any move $V_P$ of
Mr.~Paint. By induction, $G_1 - (s_1,t_1)$ and $G_2-(s_2,t_2)$ are both $\ell'$-paintable.
Analogously as in the previous case, Mrs.~Correct will play the union of her winning
strategy moves for the smaller graphs. When playing on $G$ with eraser assignment $\ell$,
whenever some of $s_1, t_1=s_2, t_2$ appear in Mr.~Paint's move, Mrs.~Correct can use one of the
two erasers assigned to $t_1=s_2$, possibly in combination with the eraser of $t_2$, in order
to guarantee that those three vertices are painted differently in case edges $(s_1,t_1)$ and
$(s_2,t_2)$ are present.
\qed
\end{proof}

%The class of series-parallel graphs with is a subclass of planar graphs defined as follows.
%It is well-known that the degeneracy of series-parallel graphs is $2$. Thus,
%Lemma~\ref{lem:degeneracy} immediately implies that series-parallel graphs are $3$-paintable.

\begin{proof}[of Lemma~\ref{lem:good-degree}]
We proceed by induction on $k$. If $k = 0$, the vertex
is isolated. Hence, it can stay coloured whenever it is
included in Mr.~Paint's move. So, assume that $k \geq 1$.
Let $G-v$ be $\ell$-paintable graph. Consider
the options how Mr.~Paint can colour $v$. If Mr.~Paint colours
set $P \subseteq V(G)$ which includes $v$ but does not include
any of its neighbours, Mrs.~Correct can leave $v$ coloured and
does not need an eraser. If $P$ includes some vertices of $v$
and Mrs.~Correct uses eraser for all of them, then $v$ can be
left coloured and Mrs.~Correct does not need to use an eraser either.
So, an eraser is needed if and only if $P$ includes some neighbours of $v$
and Mrs.~Correct's move $C$ does not use an eraser for at least one of them.
In the resulting state of game, $v$ has degree at most $k-1$ and at least $k-1$ erasers, so
the statement holds by induction. \qed
\end{proof}

\begin{proof}[of Proposition~\ref{prop:carraher}]
        \noindent (a) If Mrs.~Correct has a winning strategy on $G$, she can pursue the same strategy
        on $H$ ignoring the removed vertices and edges. As removing a structure from
        $G$ cannot invalidate any intermediate colouring created during the game, the
        strategy is winning on $H$.

        \smallskip
        \noindent (b) One implication holds by~(a): if $G$ is $\ell$-paintable, $G-v$ is its
        subgraph, so it is $\ell$-paintable too. Conversely, if $G-v$ is $\ell$-paintable,
        we consider what options of moves Mr.~Paint has with respect to $v$. If $v$ is
        included in Mr.~Paint's set which contains no neighbour of $v$, Mrs.~Correct can
        leave $v$ coloured and does not need an eraser. If some neighbours of $v$ are
        included in Mr.~Paint's move and Mrs.~Correct uses an eraser for all of them,
        then $v$ can be left coloured without using an eraser. So, erasers for $v$
        are needed only when a neighbour $w$ of $v$ belongs to Mr.~Paint's move together
        with $v$ and Mrs.~Correct does not use an eraser for $w$. However, this happens
        at most $k$-times.
\qed
\end{proof}

\end{document}